\documentclass[11pt,a4paper,reqno]{amsart}
\usepackage{amsthm,amsmath,amsfonts,amssymb,amsxtra,appendix,bookmark,dsfont,latexsym,epsfig,graphicx}

\makeatletter
\usepackage{hyperref}
\usepackage{delarray,color}
\usepackage{euscript}

\usepackage{hyperref}

\makeatletter
\def\@tocline#1#2#3#4#5#6#7{\relax
  \ifnum #1>\c@tocdepth 
  \else
    \par \addpenalty\@secpenalty\addvspace{#2}%
    \begingroup \hyphenpenalty\@M
    \@ifempty{#4}{%
      \@tempdima\csname r@tocindent\number#1\endcsname\relax
    }{%
      \@tempdima#4\relax
    }%
    \parindent\z@ \leftskip#3\relax \advance\leftskip\@tempdima\relax
    \rightskip\@pnumwidth plus4em \parfillskip-\@pnumwidth
    #5\leavevmode\hskip-\@tempdima
      \ifcase #1
       \or\or \hskip 1em \or \hskip 2em \else \hskip 3em \fi%
      #6\nobreak\relax
      \dotfill
      \hbox to\@pnumwidth{\@tocpagenum{#7}}
    \par
    \nobreak
    \endgroup
  \fi}
\makeatother

\newcommand{\bdm}{\begin{displaymath}}
\newcommand{\edm}{\end{displaymath}}
\newcommand{\bdn}{\begin{eqnarray}}
\newcommand{\edn}{\end{eqnarray}}
\newcommand{\bay}{\begin{array}{c}}
\newcommand{\eay}{\end{array}}
\newcommand{\ben}{\begin{enumerate}}
\newcommand{\een}{\end{enumerate}}

\newcommand{\R}{\mathbb{R}}
\newcommand{\C}{\mathbb{C}}
\newcommand{\cW}{\mathcal{W}}

\newcommand{\cL}{\mathcal{L}}
\newcommand{\PsiLau}{\Psi_{\rm Lau}}
\newcommand{\PsiLaun}{\Psi_{\rm Lau} ^{(\ell)}}
\newcommand{\cLau}{c _{\rm Lau}}
\newcommand{\rhoLaun}{\rho_{\rm Lau} ^{(\ell)}}

\newcommand{\cH}{\mathcal{H}}

\newtheorem{theorem}{Theorem}[section]
\newtheorem{lemma}[theorem]{Lemma}

\newtheorem{conjecture}[theorem]{Conjecture}
\theoremstyle{definition}

\theoremstyle{remark}

\newcommand{\beq}{\begin{equation}}
\newcommand{\eeq}{\end{equation}}
\newcommand{\one}{{\ensuremath {\mathds 1} }}

\numberwithin{equation}{section}

\newcommand{\LLLb}{\mathrm{LLL}_{\rm sym} ^N}
\newcommand{\LLLf}{\mathrm{LLL}_{\rm asym} ^N}

\begin{document}

\title{On the Laughlin function and its perturbations}

\author{Nicolas Rougerie}
\address{Universit\'e Grenoble Alpes \& CNRS, LPMMC (UMR 5493), B.P. 166, F-38042 Grenoble, France}
\email{nicolas.rougerie@lpmmc.cnrs.fr}





\date{June, 2019}

\begin{abstract}
The Laughlin state is an ansatz for the ground state of a system of 2D quantum particles submitted to a strong magnetic field and strong interactions. The two effects conspire to generate strong and very specific correlations between the particles. 

I present a mathematical approach to the rigidity these correlations display in their response to perturbations. This is an important ingredient in the theory of the fractional quantum Hall effect. The main message is that potentials generated by impurities and residual interactions can be taken into account by generating uncorrelated quasi-holes on top of Laughlin's wave-function. 


An appendix contains a conjecture (not due to me) that should be regarded as a major open mathematical problem of the field, relating to the spectral gap of a certain zero-range interaction.

\medskip

\noindent Expository text based on joint works with Elliott H. Lieb, Alessandro Olgiati, Sylvia Serfaty and Jakob Yngvason.
\end{abstract}

\maketitle

\tableofcontents

This text is based on a talk given at the Laurent Schwartz X-EDP seminar in Bures-sur-Yvette in December 2018 (thanks to Fran\c{c}ois Golse and Frank Merle for the invitation). The talk was mostly based on~\cite{LieRouYng-17,RouYng-17}, but the following exposition also contains results obtained later~\cite{OlgRou-19}. A conjecture, folklore in the condensed matter physics community, is presented in an appendix.

\newpage

\section{Physical motivation}

\subsection{Laughlin's state}

We are looking for a wave-function for $N$ 2D planar quantum spinless electrons. It should be a square-integrable function $\Psi_N: \R^{2N} \mapsto \C$, normalized
\begin{equation}\label{eq:mass}
 \int_{\R^{2N}} |\Psi_N| ^2 = 1
\end{equation}
and (because of the Pauli principle) antisymmetric with respect to exchange of particle labels:
\begin{equation}\label{eq:antisym}
\Psi_N (x_{\sigma(1)},\ldots,x_{\sigma(N)}) = (-1) ^{\mathrm{sign} (\sigma)} \Psi_N (x_1,\ldots,x_N) 
\end{equation}
for all $x_1,\ldots,x_N \in \R^{2}$ and all permutations $\sigma$ of the $N$ labels. 

More precisely we study the Laughlin wave-function from 1983~\cite{Laughlin-83,Laughlin-87}
\begin{equation}\label{eq:PsiLau}
 \PsiLau ^{(\ell)} (z_1,\ldots,z_N):= \cLau \prod_{1\leq i < j \leq N} (z_i-z_j) ^\ell e^{-\frac{B}{4}\sum_{j=1} ^N |z_j| ^2}
\end{equation}
where the planar coordinates $x_1,\ldots,x_N $ of the electrons are identified with complex numbers $z_1,\ldots,z_N$, $B>0$, $\ell$ is an \emph{odd}\footnote{In order to satisfy~\eqref{eq:antisym}, but even values are relevant for bosons.} integer and $\cLau=\cLau (\ell)$ is a $L^2$-normalization constant. 

As we shall explain shortly, it is natural to consider general perturbations of the Laughlin function, of the form
\begin{equation}\label{eq:PsiF}
 \Psi_F (z_1,\ldots,z_N):= c_F \PsiLaun (z_1,\ldots,z_N) F (z_1,\ldots,z_N)
\end{equation}
where $F$ is analytic in all its complex arguments and symmetric 
\begin{equation}\label{eq:sym}
F (z_{\sigma(1)},\ldots,z_{\sigma(N)}) =  F (z_1,\ldots,z_N) 
\end{equation}
for any permutation $\sigma$, so that $\Psi_F$ retains the symmetry of $\PsiLaun$. Again, $c_F$ is a $L^2$-normalization constant. 

It is of interest in fractional quantum Hall physics to restrict the general perturbation~\eqref{eq:PsiLau} to a much more restricted class
\begin{equation}\label{eq:Psif}
 \Psi_f (z_1,\ldots,z_N):= c_f \PsiLaun(z_1,\ldots,z_N) \prod_{j=1} ^N f(z_j)
\end{equation}
where $f:\C \mapsto \C$ is analytic and $c_f$ is a normalization constant. 

The purpose of this expository note is to review mathematical theorems of the following flavor:

\smallskip

\begin{multline}\label{eq:meta thm}
 \mbox{Variational problem on the full set } \{ \Psi_F \} \underset{N\to \infty}{\longrightarrow} \\ \mbox{Same variational problem, on the restricted set } \{ \Psi_f \}.
\end{multline}

\medskip

The class of ``natural'' variational problems we consider is described in more details below. The motivation is from fractional quantum Hall physics~\cite{Jain-07,Girvin-04,Goerbig-09,StoTsuGos-99,Laughlin-99}. More details can be found in the original research papers this note summarizes~\cite{LieRouYng-17,OlgRou-19,RouSerYng-13a,RouSerYng-13b,RouYng-14,RouYng-15,RouYng-17}. Other expository versions are in~\cite{LieRouYng-16,Rougerie-INSMI} and~\cite[Chapter~3]{Rougerie-hdr}. 

There are two ways, both instructive, to interpret a result of the form~\eqref{eq:meta thm}:

\medskip

\noindent \textbf{Absence of superfluous correlations}. The wave-function $\PsiLaun$ we start from contains pair correlations between electrons, because of the Jastrow factors $(z_i-z_j) ^\ell$. Since $F$ in~\eqref{eq:PsiF} is analytic, it cannot ``undo'' these correlations by canceling a factor: $\Psi_F$ can only be more correlated than $\PsiLaun$. What the statement says is that, in the situations of our concern, it is physically useless to add more correlations: the minimizers of the variational problem have an uncorrelated $F = f ^{\otimes N}$. 

\medskip

\noindent \textbf{Emergence of quasi-particles}. The analytic function $f$ defining~\eqref{eq:Psif} is essentially a polynomial. Write it as 
 \begin{equation}\label{eq:f zeroes}
 f (z) = \prod_{k=1} ^K (z-a_k) ^{m_k}. 
 \end{equation}
The complex numbers $a_1,\ldots,a_K$ are interpreted as the locations of quasi-particles (in fact, quasi-holes) generated from the Laughlin state. These are the effective particles responsible for the phenomenology of the fractional quantum Hall effect (FQHE), in particular the charge transport~\cite{SamGlaJinEti-97,MahaluEtal-97,YacobiEtal-04} in lumps of $\ell^{-1}\times $ the electron's elementary charge (the $\ell$ is that appearing in~\eqref{eq:PsiLau}). One can indeed argue  that each $a_k$ corresponds to a defect of density  $\sim m_k \ell^{-1}$ in the charge density of $\Psi_f$, as compared to the (approximately flat, see Theorem~\ref{thm:Lau} below) density of the bare $\PsiLaun$.

\medskip

We shall refer to states of the form~\eqref{eq:Psif} as the \emph{Laughlin phase}. As regards the Laughlin quasi-holes, we stress that they are here seen as classical particles, parameters in a many-body wave-function. If one considers them as true quantum particles instead, they should be be thought of~\cite{AroSchWil-84,LunRou-16} as anyons with statistics parameter $-1/\ell$. That is, they are neither bosons nor fermions, but something in between. We shall not discuss this topic here. 

\subsection{Physics background}

Let us quickly explain what $\PsiLaun$ (and its descendants) are supposed to do. In fact, it is constructed out of the following considerations: 

\medskip

\noindent \boxed{\textbf{A}} It is of the form ``analytic $\times$ gaussian'' or, in other words, made entirely of single-particle orbitals belonging to the lowest Landau level of a magnetic field perpendicular to the plane, of intensity $B$. This means that all the electrons collectively described by $\PsiLaun$ have the lowest possible magnetic kinetic energy: they do the best they can to accommodate a huge external magnetic field, the crucial ingredient of the quantum Hall effect. This is true for any choice of $\ell$ in~\eqref{eq:PsiLau}.

\medskip

\noindent \boxed{\textbf{B}} It vanishes when electrons come close, because of the zeroes of the wave-function on the hyperplanes $z_i = z_j$. The parameter $\ell$  adjusts the rate of vanishing, but it is in fact fixed by other considerations (see below). The ansatz thus seems to do a good job at reducing repulsive interactions (Coulombic, or any other repulsive interaction for that matter). 

\medskip

\noindent \boxed{\textbf{C}} The one-particle density of $\PsiLaun$ is almost constant in a thermodynamically large (i.e. of radius $\propto \sqrt{N}$) disk, with value $B/(2\pi \ell)$. The ansatz is thus relevant to describe a homogeneous electron gas at density $B/(2\pi \ell)$, in the thermodynamic limit.  

\medskip

\noindent \boxed{\textbf{D}} The function $\PsiLaun$ is expected to be very rigid in its response to perturbations. Its natural excitations are described by simple variants, the Laughlin-plus-quasi-holes states~\eqref{eq:Psif}. There are also Laughlin-plus-quasi-\emph{particles} states, but their description is not as easy.

\medskip

Points \boxed{\textbf{A}} and \boxed{\textbf{B}} conspire together: there is not much freedom in an analytic $\times$ gaussian function to vanish upon particle encounters $z_i=z_j$. The rate of vanishing is quantized (polynomial) and the zeroes come with a phase circulation/topological degree/winding number (whose chirality is imposed by the direction of the external magnetic field). 

Another way to think of the good job the wave-function does against interactions is as a kind of super-Pauli principle. While the usual Pauli principle prevents electrons from multiple occupancy of a single-particle orbital, $\PsiLaun$ has electrons occupying only $1$ out of $\ell$ available natural orbitals. These being spatially localized (again, because of the magnetic field), one can hope this favors repulsive interactions. See~\cite{CheBis-18,Haldane-11,Haldane-18,JohetalHal-16} and references therein for more details in this direction.  Also observe that, any LLL (analytic $\times$ gaussian) many-body wave function satisfying~\eqref{eq:antisym} (i.e. the Pauli principle) must be of the form~\eqref{eq:PsiF} with $\ell=1$ (because they must vanish at $z_i=z_j$). Taking a higher exponent thus enhances the Pauli principle.

\medskip

Concerning point \boxed{\textbf{C}}, we can formulate this as a theorem (see~\cite{RouSerYng-13b} and references therein for a proof; this is in fact an instance of a large class of results on classical Coulomb gases old enough for me not to try to attribute priority, see~\cite{AndGuiZei-10,Forrester-10,Serfaty-15}): 

\begin{theorem}[\textbf{Density of the Laughlin function}]\label{thm:Lau}\mbox{}\\
For a many-body wave-function $\Psi_N : \R ^{2N} \mapsto \C$ satisfying~\eqref{eq:antisym}, define its one-body density as 
\begin{equation}\label{eq:density}
\rho_{\Psi_N} (x) := N \int_{\R^{2(N-1)}} |\Psi_N (x,x_2,\ldots,x_N)|^2 dx_2 \ldots dx_N. 
\end{equation}
Let $\rhoLaun$ be the one-body density of $\PsiLaun$. Then, with $\ell$ and $B$ fixed, 
\begin{equation}\label{eq:dens Lau}
 \int_{\R^2} \left| \rhoLaun - \frac{B}{2\pi \ell} \one_{D(0,R)}\right| \ll N  
\end{equation}
in the limit $N\to \infty$. Here $D(0,R)$ is the disk of center $0$ and radius 
$$ R = \sqrt{\frac{2\ell N}{B}}.$$
\end{theorem}

Points \boxed{\textbf{A}}, \boxed{\textbf{B}} and \boxed{\textbf{C}} above indicate (with a bit of hand-waving, perhaps) that $\PsiLaun$ should be a very good ansatz for minimizing the energy of a 2D electron gas in a uniform perpendicular magnetic field, at density\footnote{I.e. at filling factor $\nu = (hc/e) (\rho/B) = \ell ^{-1}$ in view of our implicit choice of units $\hbar = c = e = 1$.} $\rho = B (2\pi \ell) ^{-1}$. This fact has been extensively checked numerically, see~\cite{Jain-07,Girvin-04,Goerbig-09,StoTsuGos-99,Laughlin-99} for references.

\medskip

What about point \boxed{\textbf{D}} ? There are two aspects to it. One is a major open problem, that one can refer to as the \emph{spectral gap conjecture}. I have nothing new to report on it, see Appendix~\ref{app:spectral gap} for a precise formulation. The other aspect is the topic of the following exposition, so see below. 

To conclude with the physical motivation, we should emphasize the importance of point \boxed{\textbf{D}} for the theory of the FQHE. Indeed, the effect does \emph{not} occur in a \emph{homogeneous electron gas at zero temperature and density $\rho = B (2\pi \ell) ^{-1}$} but in an \emph{electron gas with impurities at small temperature and density in the vicinity of $\rho = B (2\pi \ell) ^{-1}$}. This is not a technical distinction: without impurities there would be no effect, and the essence of the QHE is a quantized plateau in the Hall conductivity of the sample for densities \emph{close to} $B /(2\pi \ell)$. The latter fact is crucial to the main application of the QHE (in metrology, setting the standard for the von Klitzing constant $h/e^2$). Again, see~\cite{Jain-07,Girvin-04,Goerbig-09,StoTsuGos-99,Laughlin-99} for introductions to the topic.

\section{The mathematical problem and the main theorem}

We turn to proposing a mathematical formulation for point \boxed{\textbf{D}} of the previous section. We first observe that the (more or less informal) arguments \boxed{\textbf{A}} and \boxed{\textbf{B}} proposed to arrive at Laughlin's wave-function~\eqref{eq:PsiLau} in fact point to the larger class of functions~\eqref{eq:PsiF} built on it. So why indeed restrict attention to~\eqref{eq:PsiLau}, or even to the functions in~\eqref{eq:Psif} decorated by quasi-holes ? \underline{Argument 1}: why should we do something complicated when we can try something simpler first ? \underline{Argument 2}: by fiddling around the proof of Theorem~\ref{thm:Lau} and/or more informal intuitions, one can guess that $\PsiLau$ will be the only function of the class~\eqref{eq:PsiF} with mean density $B/(2\pi \ell)$. But we have explained (or at least, alluded to) the fact that in FQHE physics it is crucial to understand what happens around the special density $B/(2\pi \ell)$, in particular, for smaller values. 

The problem we propose below is intended to shed some light on these issues. It takes for granted the restriction of admissible states to the class~\eqref{eq:PsiF}. It is believed that the class~\eqref{eq:PsiF} is an approximate low-energy eigenspace for the full physical Hamiltonian, separated by a gap (which does not close in the thermodynamic limit) from the rest of the spectrum.  In Appendix~\ref{app:spectral gap}, I present a problem in this direction: the class~\eqref{eq:PsiF} is an exact ground eigenspace for an approximate Hamiltonian, is it separated by a gap from the rest of the spectrum ?

\medskip

In the spirit of degenerate perturbation theory, we consider the problem of minimizing what is left of the energy, within the class~\eqref{eq:PsiF}. Since the magnetic kinetic energy is fixed by the restriction to lowest Landau level (analytic $\times$ gaussian) functions, all that is left to minimize are the interaction energy and the energy due external potentials (trapping and/or impurities). We thus arrive at the following problem 
\begin{equation}  \label{eq:qm_energy}
E (N,\lambda)=\inf\Big\{\mathcal{E}_{N,\lambda}[\Psi_F]\;|\;\Psi_F \mbox{ of the form~\eqref{eq:PsiF}},\,\int_{\mathbb{R}^{2N}}|\Psi_F|^2=1\Big\}.
\end{equation}
where 
\begin{equation} \label{eq:many_body_energy}
\mathcal{E}_{N,\lambda}[\Psi_F]=\Big\langle\Psi_F\Big|\sum_{j=1}^N V(x_j)+\lambda \sum_{i<j}W(x_i-x_j)\Big|\Psi_F\Big\rangle_{L^2}.
\end{equation}
Here $V,W:\R^2 \mapsto \R$ are the external and pair-interaction potentials respectively (seen as mutliplication operators in~\eqref{eq:many_body_energy}), and $\lambda \in \R$ is a coupling constant. Note that $W$ might be reduced a lot (this is in fact hoped for) by restricting to the class $\Psi_F$, but it is not strictly canceled (unlike the zero-range interactions considered in Appendix~\ref{app:spectral gap}), in particular it can have a long-range, 3D-Coulomb-like, part.

A more precise version of the statement~\eqref{eq:meta thm} is as follows. Consider the simpler energy
\begin{equation} \label{eq:qh_energy}
e (N,\lambda)=\inf\Big\{\mathcal{E}_{N,\lambda}[\Psi_f]\;|\;\Psi_f \text{ of the form \eqref{eq:Psif}},\,\int_{\mathbb{R}^{2N}}|\Psi_f|^2=1\Big\}.
\end{equation}
Obviously $E(N,\lambda) \leq e(N,\lambda)$. We would like to prove that 
\begin{equation} \label{eq:anticipation}
\boxed{E (N,\lambda)\simeq e (N,\lambda)\quad\text{as}\;N\to\infty \mbox{ with } \lambda \mbox{ fixed}.}
\end{equation}
This means that, for the purpose of minimizing a natural energy functional, it is sufficient to restrict to the sub-class~\eqref{eq:Psif} instead of considering the fully general~\eqref{eq:PsiF}. We can prove this under some simplifying assumptions that we now describe.

Since, in view of Theorem~\ref{thm:Lau}, the Laughlin state lives on thermodynamically large length scales $\propto \sqrt{N}$, it is natural to demand that the potentials $V$ and $W$ also do. We thus set, for fixed functions $v,w$,
\begin{equation}
V (x) = v\left(N^{-1/2} x\right)
\end{equation}
and (the $N$ pre-factor ensures that the potential and interaction energies stay of the same order when $N\to \infty$)
\begin{equation} \label{eq:rescaled_w}
W (x) = N ^{-1} w \left(N^{-1/2} x\right).
\end{equation}
Note that it \emph{is} physically relevant to consider potentials living on smaller length scales, in particular to take impurities into account. We thus simplify the physics of the problem at this point. 

In~\cite{OlgRou-19} we proved the $\lambda \neq 0$ version of the following theorem, while the (still highly non-trivial) $\lambda = 0$ case was solved earlier~\cite{LieRouYng-17,RouYng-17}:

\begin{theorem}[\textbf{Energy of the Laughlin phase}]\label{thm:ener}\mbox{}\\
Assume that $v$ and $w$ are smooth fixed functions. Assume that $v$ goes to $+\infty$ polynomially at infinity, and that it has finitely many non-degenerate critical points. There exists $\lambda_0 >0$ such that
$$ \frac{E(N,\lambda)}{e(N,\lambda)} \underset{N\to \infty}{\to} 1$$
with $B > 0$, $\ell$ a fixed odd integer and $|\lambda| \leq \lambda_0$. 
\end{theorem}

\begin{proof}[Comments] \mbox{}\\
\noindent \textbf{1}. The assumption that $|\lambda|$ be small enough is probably necessary. Indeed, increasing $\lambda \geq 0$ would mean decreasing the relative influence of the external potential $V$, which in particular represents trapping. Less trapping should mean a lower mean density. But for significantly lower densities, the likely behavior of the system is not to generate more quasi-holes as in~\eqref{eq:Psif} but to reorganize into a different fractional quantum Hall state, e.g. a Laughlin wave-function with higher exponent. What the theorem shows is that for $\lambda$ small but $O(1)$ in the $N\to \infty$ limit, the system does respond by generating quasi-holes. This robustness is related to the finite width of the plateaus in the Hall conductivity, although a full explanation thereof requires many more ingredients. 

\medskip

\noindent \textbf{2}. The smoothness assumptions on the potentials should not be necessary, although our method of proof does demand some regularity. In particular, for the Coulomb interaction $w (x) = |x| ^{-1}$ the theorem is inconclusive.

\medskip

\noindent \textbf{3}. A key tool in the proof is to relate the two infima~\eqref{eq:many_body_energy}-\eqref{eq:qh_energy} to the flocking~\cite{BurChoTop-15,FraLie-16} (or, for $\lambda = 0$, bath-tub~\cite[Theorem~1.14]{LieLos-01}) energy
\begin{equation}\label{eq:flocking}
E^{\rm flo} (N,\lambda) := \inf \left\{\int_{\R^2} V \varrho + \frac{\lambda}{2} \iint_{\R^2} \varrho(x) W (x-y) \varrho(y) dxdy  , \: 0 \leq \varrho \leq \frac{B}{2\pi \ell}, \: \int_{\R^2} \varrho = N \right\}. 
\end{equation}
In fact, since $E(N,\lambda) \leq e (N,\lambda)$, it is sufficient to prove 
$$ e (N,\lambda) \lessapprox E^{\rm flo} (N,\lambda)$$
by a trial state argument, and 
\begin{equation}\label{eq:low bound}
 E(N,\lambda) \gtrapprox E^{\rm flo} (N,\lambda), 
\end{equation}
which, as for most variational problems, is the hardest inequality.

\medskip

\noindent \textbf{4}. In order to prove~\eqref{eq:low bound}, the main tool is what we coined an \emph{incompressibility estimate}, starting from~\cite{RouYng-14}. Namely, for any $L^2$-normalized $\Psi_F$ of the form~\eqref{eq:PsiF}, with $\rho_F$ the associated one-particle density~\eqref{eq:density}, we have in an appropriate sense 
\begin{equation}\label{eq:incomp}
\rho_F \leq \frac{B}{2\pi \ell} (1 + o_N(1)). 
\end{equation}
Note that this holds irrespective of the (sequence of) analytic factor(s) $F$ chosen, so that the variational set defining~\eqref{eq:many_body_energy} is in some sense included in that defining~\eqref{eq:flocking}.  

\medskip

See~\cite{LieRouYng-17,OlgRou-19,RouSerYng-13a,RouSerYng-13b,RouYng-14,RouYng-15,RouYng-17} for more details, and for corollaries regarding the minimizers of the problem.
\end{proof}

\medskip

In the next section I discuss in more details the tool~\eqref{eq:incomp}. Partial results were obtained in~\cite{RouYng-14,RouYng-15}, and a satisfactory statement (although there is still room for improvement) in~\cite{LieRouYng-16,LieRouYng-17}. For the results of~\cite{OlgRou-19} we in fact need something stronger than~\eqref{eq:incomp}: the latter is a result in expectation, for the $\lambda \neq 0$ case we need a deviation result.

\section{Incompressibility estimates for 2D Coulomb systems}

Let me now present the main insight behind~\eqref{eq:incomp}, and some elements of proof.

\subsection{Plasma analogy}

The main way one has to get to grips with Laughlin's wave-function and its descendants is via an analogy with classical statistical mechanics, more precisely the 2D one-component plasma, or log-gas, or $\beta$-ensemble, which is itself connected to random matrices via the Ginibre ensemble~\cite{AndGuiZei-10,Forrester-10,Mehta-04,Serfaty-15}. The analogy originates in the seminal paper~\cite{Laughlin-83}.

For the applications we have in mind, one is only interested in the probability density $|\Psi_F| ^2$ of a function $\Psi_F$ of form~\eqref{eq:PsiF}. One writes it as a Boltzmann-Gibbs factor,
\begin{equation}\label{eq:plasma}
 |\Psi_F (z_1,\ldots,z_N)| ^2 = \frac{1}{\mathcal{Z}_F} \exp\left( - H_F (z_1,\ldots,z_N)\right)
\end{equation}
with $\mathcal{Z}_F$ ensuring $L^1$-normalization (partition function of the effective plasma) and $H_F$ an effective Hamilton function
\begin{equation}\label{eq:class hamil}
H_F (z_1,\ldots,z_N) = \frac{B}{2} \sum_{j=1} ^N |z_j| ^2 - 2 \ell \sum_{1\leq i < j \leq N} \log |z_i-z_j| - 2 \log \left| F (z_1,\ldots,z_N) \right|. 
\end{equation}
Usually, writing a function as the exponential of its logarithm is not a particularly brilliant idea. The reason it is in this case is that, in a 2D universe, the function $H_F$ has a clear interpretation\footnote{This effective Hamiltonian is a technical tool, it has nothing to do with the original, physical, energy that $\Psi_F$ should be chosen to minimize.} in terms of electrostatics. This is because in 2D the electrostatic potential $\Phi$ generated by a charge distribution $\sigma$ is given by 
$$ - \Delta \Phi = \sigma, \quad \Phi = -\frac{1}{2\pi}\log|\,.\,| \star \sigma.$$
More precisely $H_F$ is the energy of $N$ mobile 2D particles (at locations $z_1,\ldots,z_N \in \C \leftrightarrow \R^2$) of charge $- \sqrt{4\pi \ell}$

\medskip 

\noindent \textbf{1}. Interacting among themselves via repulsive Coulomb forces. 

\medskip
 
\noindent \textbf{2}. Attracted to a fixed uniform background of charge density $+B \ell ^{-1/2} / \sqrt{\pi}$.

\medskip

\noindent \textbf{3}. Feeling the potential $\cW = -2 \log |F|$ generated by additional ``phantom'' $+$ charges. The location of the latter can be essentially arbitrary, and correlated with the positions of $z_1,\ldots,z_N$, but their charge must be positive because 
$$ - \Delta_{z_j} \cW  \geq 0$$
for any $j$ (recall that $F$ is analytic).

\medskip

Now the bound~\eqref{eq:incomp} ought to become less mysterious. The density $\rho_F$ is interpreted as the mean distribution of the charges in the plasma described above, at thermal equilibrium with temperature $1$. It turns out that this is an effectively small temperature, so that one can just as well think of the charges as distributed to minimize the energy given by $H_F$.  

Then, the density value $B/(2\pi \ell)$ appearing in~\eqref{eq:incomp} is just that corresponding to local charge neutrality for the effective plasma (compare points \textbf{1} and \textbf{2} above), a situation notoriously favorable to minimize the electrostatic energy. This is the intuitive explanation behind Theorem~\ref{thm:Lau}: for $F= 1$ there are only the effects of \textbf{1} and \textbf{2} to take into account, so that the density wants to be \emph{equal} to $B/(2\pi \ell)$. This can be made very precise~\cite{BauBouNikYau-15,BauBouNikYau-16,Leble-15b,LebSer-16}.

What about the effect of the additional potential $\cW = - 2 \log |F|$ ? We know essentially nothing of it, except that it is generated by a positive charge distribution and hence exercises a \emph{repulsive} force on the points $z_1,\ldots,z_N$. Hence it is perhaps natural that a non-trivial analytic $F$ leads to a smaller density, whence~\eqref{eq:incomp}. However this is much more subtle than it looks. The bound~\eqref{eq:incomp} should hold \emph{everywhere} in space. Why cannot one generate a local bump of charge above the preferred value $B/(2\pi \ell)$ by acting suitably with repulsive charges ?    

Consider the following thought experiment: we are given a patch of negative (mobile) and positive (fixed) charges, screening one another (i.e. the total charge distribution is locally neutral). Now add around this patch additional (phantom) positive charges (generating $\cW$), pushing on the negative charges already present. The kind of result we aim at is: however we distribute the phantom charges, the effect is that the mobile negative charges leak out the original patch, but never accumulate above the density of the fixed charges. 

Even worse: since the positions of the phantom charges can (via the zeroes of the analytic function $F$) be correlated in any way we like with the positions of the mobile charges, one could even make them ``run after'' the leaking charge. If there is some leaking, the phantom charges relocate themselves so as to push back the indisciplined mobile charges to force them to concentrate. This, we say, can only result in more leaking.  

That the bound~\eqref{eq:incomp} is true is a remarkable instance of the power of screening in electrostatics. We give a sketch of the main part of the proof of~\eqref{eq:incomp} in the next subsections.

\subsection{Density bounds for Coulomb ground states}

In this exposition we shall be content with explaining why the density bound~\eqref{eq:incomp} is true at the level of ground states of~\eqref{eq:class hamil}. Taking into account the temperature is done in a second step, and is the less optimal part of the proof as it now stands. From now one we thus focus exclusively on the effective 2D electrostatic problem described in the previous subsection.

Let us clean the notation a bit. By changing length and energy units we can consider the Hamilton function 
\begin{equation}\label{eq:class hamil 2}
\cH (x_1,\ldots,x_N) = \frac{\pi}{2} \sum_{j=1} ^N |x_j| ^2 - \sum_{1\leq i < j \leq N} \log |x_i - x_j| + \cW (x_1,\ldots,x_N) 
\end{equation}
with $x_1,\ldots,x_N \in \R ^2$ and $\cW$ superharmonic in each variable:
\begin{equation}\label{eq:superharm}
 -\Delta_{x_j} \cW \geq 0, \quad \forall j. 
\end{equation}
We consider only zero-temperature equilibrium configurations (minima of $\cH$) and want to prove that their density of points is everywhere bounded above by $1$ (the neutrality density in the new units). The following, proved in~\cite{LieRouYng-17}, shows that this is true on any length scale much larger than the typical inter-particle distance ($O(1)$ independently of $N$ in the new units): 

\begin{theorem}[\textbf{Incompressibility for 2D Coulomb ground states}]\label{thm:incomp GS}\mbox{}\\
There exists a bounded function $g:\R^+ \mapsto \R^+$, independent of $N$ and $\cW$, with 
$$
g(R) \underset{R\to \infty}{\to} 0,
$$ 
such that, for any $X_N^0 = (x_1^0,\ldots,x_N ^0)$ minimizing $\cH$, any point $a\in \R^2$ and any radius $R > 0$
\begin{equation}\label{eq:incomp GS}
N(a,R) := \sharp\left\{ x_j^0 \in X_N ^0 \cap D(a,R) \right\} \leq \pi R^2 (1 + g(R)) 
\end{equation}
where $D(a,R)$ is the disk of center $a$ and radius $R$ and $\sharp$ stands for the cardinal of a discrete set. 
\end{theorem}

\begin{proof}[Comments] 
A look at a simplified problem is instructive. Consider a positive measure $\sigma$ with 
$$ \int_{\R^2} \sigma = N,$$
the particle number, minimizing a continuous/mean-field version of the energy~\eqref{eq:class hamil 2}:
\begin{multline}\label{eq:class MF}
 \frac{\pi}{2}  \int_{\R^2} |x| ^2 \sigma (x) dx - \frac{1}{2} \iint_{\R^2 \times \R^2} \sigma (x) \log |x-y| \sigma (y) dx dy \\
 + N^{-N}\int_{\R^{2N}} \cW (x_1,\ldots,x_N) \sigma(x_1) \ldots \sigma(x_N) dx_1 \ldots dx_N. 
\end{multline}
The Euler-Lagrange for this problem says that, on the support of $\sigma$, 
$$ \frac{\pi}{2} |x| ^2 - \log |\,.\,| \star \sigma (x)+ N^{1-N} \int_{\R^{2(N-1)}} \cW (x,x_2,\ldots,x_N) \sigma(x_2) \ldots \sigma(x_N) dx_2 \ldots dx_N = \mu,$$ 
a constant (Lagrange multiplier associated with the mass constraint). Taking the Laplacian of the above equation and using~\eqref{eq:superharm} immediately gives 
$$ \sigma \leq 1.$$
The issue is that, for a general genuine many-body potential $\cW$, it does not seem feasible to reduce the minimization of~\eqref{eq:class hamil 2} for point configurations to that of~\eqref{eq:class MF} for measures. For particular $\cW$ containing only few-particle interactions one can pass rigorously~\cite{RouYng-14} from~\eqref{eq:class hamil 2} to~\eqref{eq:class MF}, which gives a particular, weaker, case of the Theorem and its applications. The proof of the general case, sketched below, does not use the continuum/mean-field approximation.
\end{proof}

\subsection{Sketch of proof for Theorem~\ref{thm:incomp GS}}

We present the four main lemmas of the proof and a few explanations of how they fit together. See~\cite{LieRouYng-16,LieRouYng-17} for more details. The method is rooted in potential theory.

 \begin{figure}[hb]
\begin{center}
\includegraphics[width=10cm]{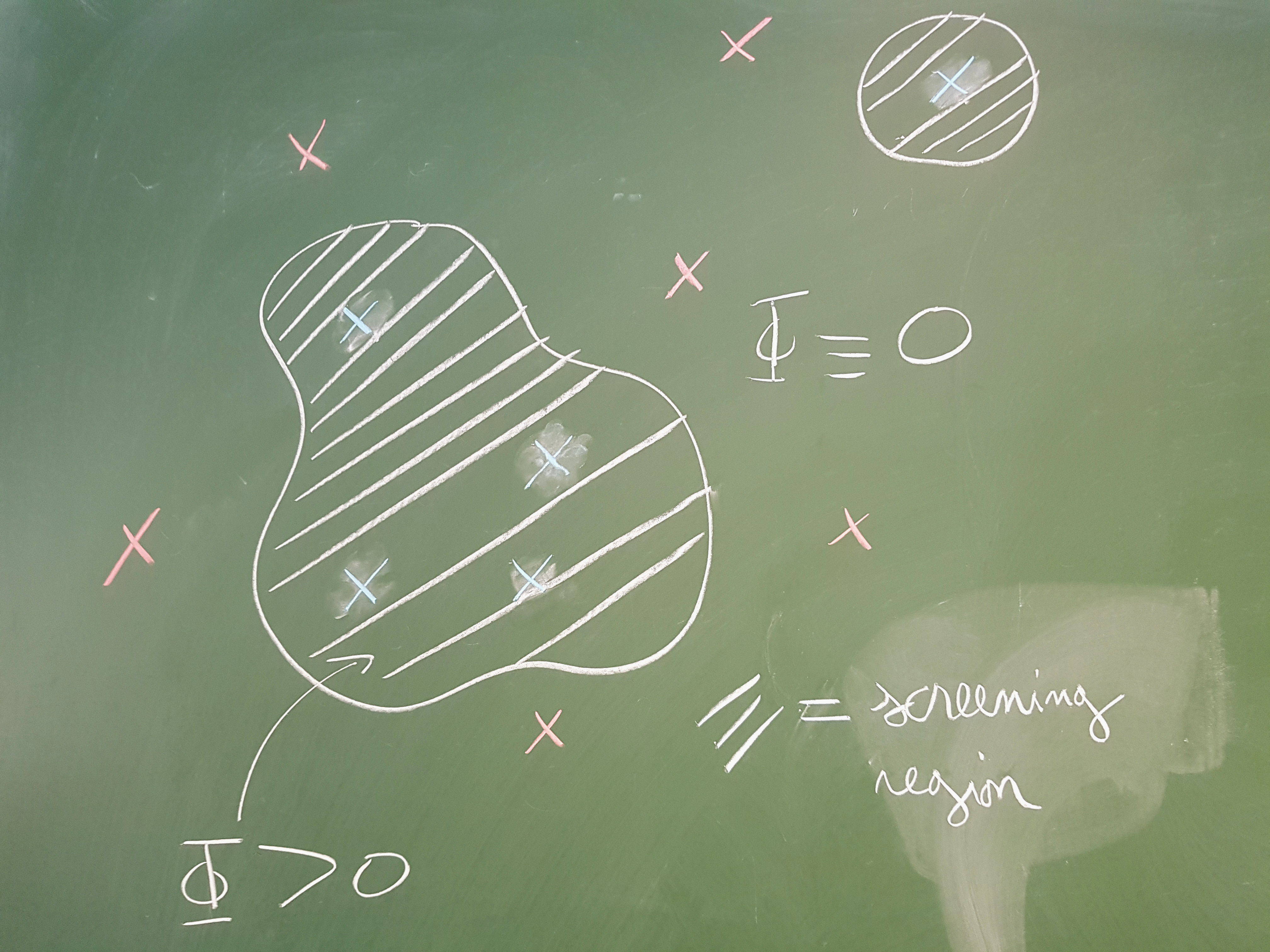}\\
%
%
%
\caption{\textbf{Exclusion rule}. The blue points generate a screening region (dashed white). No other (red) point of a minimizing configuration may lie within it.}
\label{fig:exclusion}
\end{center}
 \end{figure}

\begin{lemma}[\textbf{Screening regions}]\label{lem:screening}\mbox{}\\
 Let $x_1,\ldots,x_K$ be points in $\R^2$. There exists an open set $\Sigma = \Sigma(x_1,\ldots,x_K) \subset \R ^2$ with Lebesgue measure 
\begin{equation}\label{eq:size}
 \left| \Sigma(x_1,\ldots,x_K) \right| = K 
\end{equation}
such that the electrostatic potential 
\begin{equation}\label{eq:potential}
 \Phi := -\log |\,.\,| \star \left( \sum_{k=1} ^K \delta_{x_j} - \one_{\Sigma} \right) 
\end{equation}
satisfies 
\begin{equation}\label{eq:PhiSigma}
\begin{cases}
\Phi &> 0 \mbox{ almost everywhere in } \Sigma\\  
\Phi &= 0 \mbox{ almost everywhere in the complement of } \Sigma. 
\end{cases}
\end{equation}
\end{lemma}

\begin{proof}[Comments]
In other words, given $K$ charges in $\R ^2$, one can always screen their electrostatic potential by putting a patch of constant charge density of opposite sign around them. The proof is semi-constructive, using an auxiliary variational problem (incompressible neutral Thomas-Fermi molecule) whose solution is the indicative function of $\Sigma$.
\end{proof}

The utility for the minimizers of~\eqref{eq:class hamil 2} is as follows (see Figure~\ref{fig:exclusion} for illustration):

\begin{lemma}[\textbf{Exclusion rule}]\label{lem:exclusion}\mbox{}\\
Let $X_N^0$ be a minimizing configuration of points for~\eqref{eq:class hamil 2}. Let $K < N$ be an integer and $y_1,\ldots,y_{K+1} \in X_N ^0$. Then, $y_{K+1}\notin \Sigma (y_1,\ldots,y_K)$, where  $\Sigma (y_1,\ldots,y_K)$ is the screening region of Lemma~\ref{lem:screening}. We refer to this as the \emph{exclusion rule}.    
\end{lemma}

\begin{proof}
For minimality, $y_{K+1}$ must minimize $\cW$ plus the potential generated by all the other points and the background. If $y_{K+1}\in \Sigma (y_1,\ldots,y_K)$, one would be able to decrease this total potential by moving it the boundary of $y_{K+1}$. The argument uses~\eqref{eq:PhiSigma} and the superharmonicity of $\cW$.
\end{proof}

Now we can forget about the minimality of the configuration $X_N^0$, for we have the 

\begin{lemma}[\textbf{Exclusion $\Rightarrow$ density bound}]\label{lem:density}\mbox{}\\
There exists a continuous function $g:\R^+ \mapsto \R^+$ going to $0$ at infinity such that, for any (possibly infinite) point configuration $X = x_1,\ldots,x_N,\ldots $ satisfying the exclusion rule of Lemma~\ref{lem:exclusion} and for any disk $D(a,R)$ we have 
\begin{equation}\label{eq:dens excl}
\sharp\left\{ x_j \in X \cap D(a,R)\right\} \leq \pi R^2 (1+g(R)).
\end{equation}
\end{lemma}

\begin{proof}
Consider a minimizing configuration, and a disk $D(a,R)$. The proof is illustrated by Figure~\ref{fig:good}, to which the color code below refers. Suppose the screening region (white line) from Lemma~\ref{lem:screening} generated by the (blue) points inside the (blue) circle is included within a slightly larger circle (dashed white) of radius $\tilde{R}$. Then we know from~\eqref{eq:size} that the number of points in the disk $D(a,R)$ is not larger than $\pi \tilde{R} ^2$. If we can construct such a $\tilde{R}\sim R$ when $R\to \infty$, we have won. Since the (white) screening region may not contain any of the (red) points outside the disk, it is conceivable that the latter force the white line to stay close to the blue circle. The typical enemy we fight is a pathological configuration such as that presented in Figure~\ref{fig:bad}. The screening region sends long serpentine tendrils to infinity, still avoiding the red points. Its area is still (by definition) the number of points in the disk, but the latter can now have an area much smaller than that of the screening region.  This kind of pathology is excluded by Lemma~\ref{lem:support} below.  
\end{proof}

 \begin{figure}[h]
\begin{center}
\includegraphics[width=10cm]{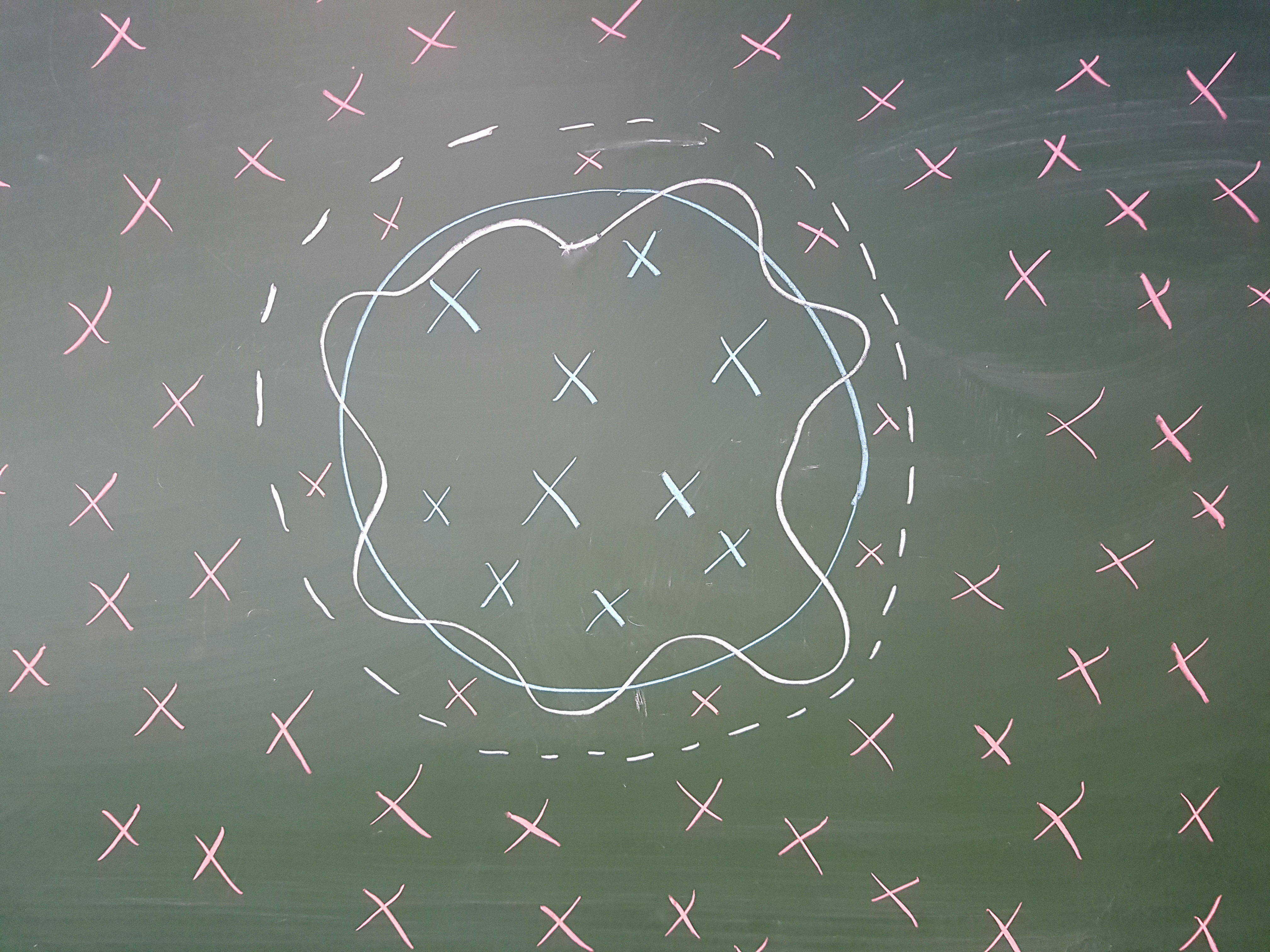}\\
%
%
%
\caption{\textbf{Good configuration}. The blue points inside the (blue) circle generate a screening region (inside of white line), avoiding the other (red) points. It is contained in a disk (inside of white dashed circle) not too large compared to the original blue circle.}
\label{fig:good}
\end{center}
 \end{figure}

 \begin{figure}[h]
\begin{center}
\includegraphics[width=10cm]{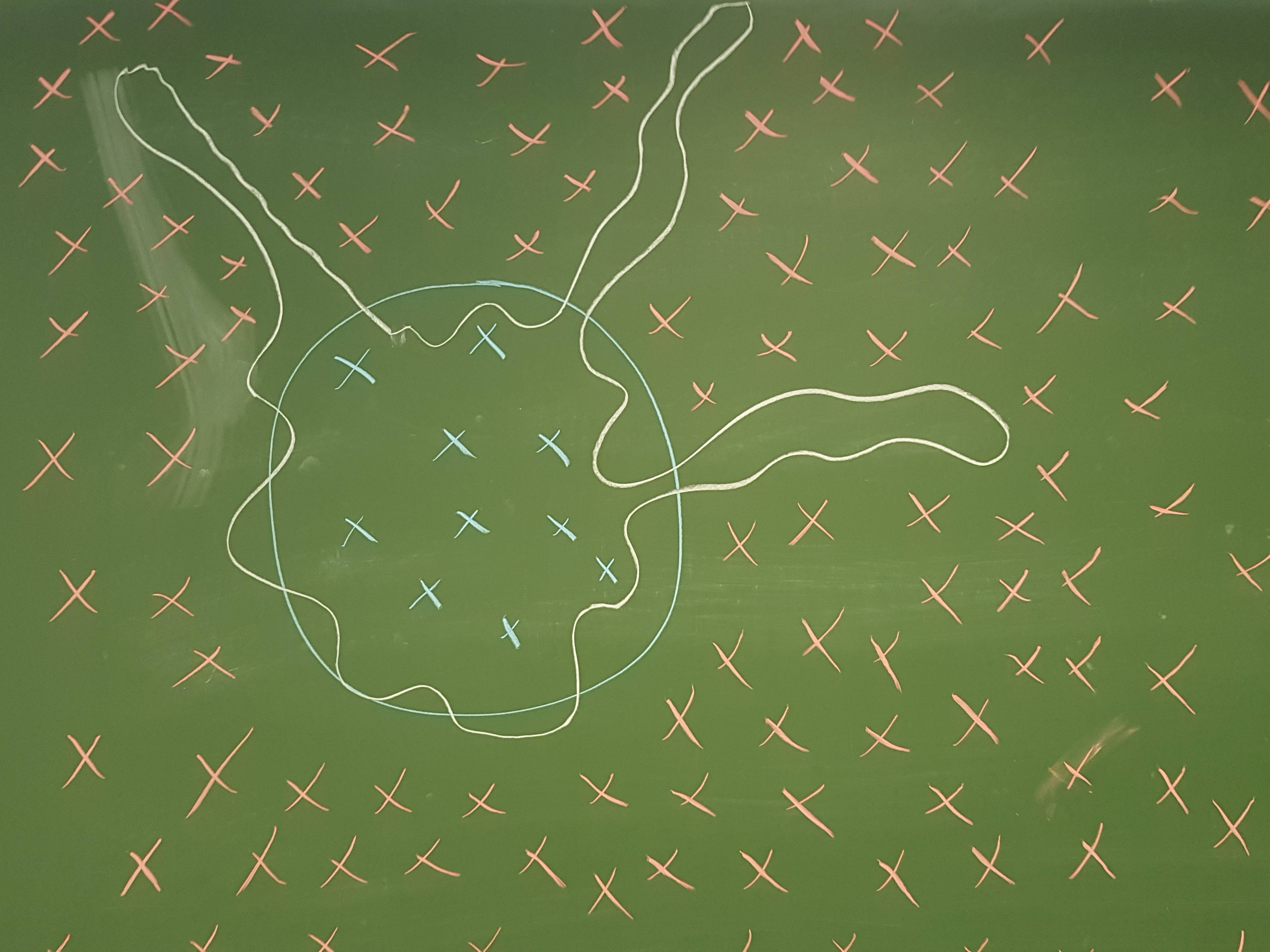}\\
%
%
%
\caption{\textbf{Pathological configuration, to be excluded}. The blue points inside the (blue) circle generate a screening region (inside of white line). It sends tendrils out to infinity, while still avoiding the other (red) points.}
\label{fig:bad}
\end{center}
 \end{figure}

Clearly Theorem~\ref{thm:incomp GS} follows from the previous lemmas. The final element of the proof of Lemma~\ref{lem:density} is 

\begin{lemma}[\textbf{Support of the screening region}]\label{lem:support}\mbox{}\\
We use the notation of Lemma~\ref{lem:support}. There exists a constant $C>0$ such that, for any $a\in \R^2$ and $r \in \R^+$
$$ \Sigma \subset D(a,R)$$
with 
$$ R = r + C \sqrt{\max\left\{ |\Phi (x)|, \: |x-a| = r\right\}}.$$
\end{lemma}

\begin{proof}[Comments]
In other words, if one happens to know that the potential~\eqref{eq:potential} is small on some circle, then the screening region must sit inside a slightly larger concentric circle: there is no need to screen any further to make the potential vanish.  
 
To see that this completes the proof of Lemma~\ref{lem:density} we argue as follows. One can assume the density of point in the configuration is bounded below everywhere (this requires some argument, but clearly, a high density everywhere is the most likely enemy). Recall that (blue, red etc ... again refer to the color-code of Figures~\ref{fig:good} and \ref{fig:bad}) the potential~\eqref{eq:potential} generated by the blue points and the screening region must vanish at all the red points. Since there are many such points outside the blue circle, it takes only a small leap of faith (and/or a few estimates) to hope that the potential must in fact be small uniformly outside of the blue circle. But then Lemma~\ref{lem:support} implies that the screening region (white line) must be included in a slightly larger disk (dashed line in Figure~\ref{fig:good}). The pathological configuration of Figure~\ref{fig:bad} is thus excluded. 
\end{proof}

\newpage

\appendix

\section{The spectral gap conjecture}\label{app:spectral gap} 

Here I expose a conjecture whose resolution would go a long way towards a full rigorous justification of point \boxed{\textbf{D}} of the introduction. The conjecture is not mine: it can be traced back to the fundamental papers~\cite{Laughlin-83,Haldane-83}, and is more or less folklore in the condensed matter physics community. I am grateful to F.D.M. Haldane, E. H. Lieb and J. Yngvason in particular for discussions relating to the topic below. Previous explicit mentions of the conjecture are e.g. in~\cite{LewSei-09,RouSerYng-13b}.

\medskip 

To obtain a clean mathematical statement, we consider a toy Hamiltonian defined as follows. Let the bosonic and fermionic lowest Landau levels be respectively
\begin{align}
\LLLb &= \left\{ A(z_1,\ldots,z_N) e^{-\frac{B}{4} \sum_{j=1} ^N |z_j| ^2 }, \quad A \mbox{ analytic and symmetric} \right\} \\
\LLLf &= \left\{ A(z_1,\ldots,z_N) e^{-\frac{B}{4} \sum_{j=1} ^N |z_j| ^2 }, \quad A \mbox{ analytic and antisymmetric} \right\}
\end{align}
where symmetric/antisymmetric means ``under exchange of the labels of the coordinates $z_1,\ldots,z_N$''. On these spaces, consider the $m$-th Haldane pseudo-potential Hamiltonian 
\begin{equation}\label{eq:pseupot}
H (m,N) := \sum_{1\leq i < j \leq N} |\varphi_m \rangle \langle \varphi_m |_{ij} 
\end{equation}
where $|\varphi_m \rangle \langle \varphi_m |_{ij}$ projects the relative coordinate\footnote{Again, $\R^2 \ni x \leftrightarrow z \in \C$.} $x_i - x_j$ of particles $i$ and $j$ on the one-body state ($c_m$ is a normalization constant)
$$\varphi_m (z) = c_m z^m e^{-\frac{B}{4} |z| ^2}.$$
Note that, when acting on $\LLLb$ or $\LLLf$, only for even (respectively, odd) $m$ does $H(m,N)$ act non-trivially.

Clearly, $\PsiLaun$ is an exact ground state, i.e. eigenfunction with eigenvalue $0$ for $H(\ell-2,N)$ acting on $\LLLb$ (even $\ell$) or $\LLLf$ (odd $\ell$). The conjecture says that the gap above the eigenvalue $0$ does not close in the thermodynamic limit $N\to \infty$. To formulate it, observe first that 
$H(\ell-2,N)$ commutes with the total angular momentum operator 
\begin{equation}\label{eq:tot mom}
\cL_N := \sum_{j=1} ^N z_j \partial {z_j} - \overline{z}_j \partial_{\overline{z}_j}, 
\end{equation}
and consider a joint diagonalization of the two operators on either $\LLLb$ (if $\ell$ is even) or $\LLLf$ (if $\ell$ is odd). The angular momentum of the Laughlin state~\eqref{eq:PsiLau} is
$$ \cL_ N \PsiLaun = \frac{\ell}{2} N (N-1) \PsiLaun.$$

\begin{conjecture}[\textbf{Spectral gap conjecture}]\label{conj:spec}\mbox{}\\
Consider the spectral gap of $H_{\ell-2,N}$ on the sector of angular momenta below that of the Laughlin state
\begin{equation}\label{eq:spectral gap}
 \sigma (N,\ell) = \inf \left\{ \mathrm{spec} \left(H(\ell-2,N) _{|\cL_N \leq \frac{\ell}{2} N (N-1)} \right) \setminus \{ 0 \}\right\}.
\end{equation}
There exists a constant $c_{\ell}>0$, independent of $N$, such that 
$$ \sigma (N,\ell) \geq c_\ell >0.$$
\end{conjecture}

To motivate the conjecture, observe that if one projects a bona-fide pair interaction Hamiltonian
$$ H_w = \sum_{1\leq i< j \leq N} w (x_i-x_j)$$
with radial potential $w\geq 0$ on the LLL, one obtains 
\begin{equation}\label{eq:LLL hamil}
H_w ^{\mathrm{LLL}} := P_{\mathrm{LLL}^N_{\rm sym/asym}} H_w \: \: P_{\mathrm{LLL}^N_{\rm sym/asym}}= \sum_{i<j} \sum_{m\geq 0} \left\langle \varphi_m | w | \varphi_m \right\rangle |\varphi_m \rangle \langle \varphi_m |_{ij}. 
\end{equation}
The coefficients $\left\langle \varphi_m | w | \varphi_m \right\rangle$ are called ``Haldane pseudo-potentials''. The toy Hamiltonian~\eqref{eq:pseupot} above is obtained by discarding all terms from the sum but one, in order for the Laughlin state to be an exact ground state, and not just a very good approximation. 

There is one particular case, namely $\ell=2$, with $H (0,N)$ acting on $\LLLb$ where this truncation of~\eqref{eq:LLL hamil} is more than a crude simplification. Indeed, $H_{0,N}$ is nothing but a Dirac-delta interaction projected on the LLL. This makes perfect sense~\cite{BerPap-99,RouSerYng-13b,LieSeiYng-09,LewSei-09,PapBer-01,GirJac-84} since LLL functions are very regular. In fact $H_{0,N}$ acts as 
\begin{multline*}
 H_{0,N} \left( A (z_1,\ldots,z_N) e ^{-\frac{B}{4} \sum_{j=1} ^N |z_j|^2} \right) = \\ 
 \frac{1}{2\pi}\sum_{1\leq i <j\leq N} A \left(z_1,\ldots, \frac{z_i+z_j}{2},\ldots, \frac{z_i+z_j}{2},\ldots ,z_N\right) e ^{-\frac{B}{4} \sum_{j=1} ^N |z_j|^2}
\end{multline*}
and this model can be derived from a true many-body Hamiltonian in a physically relevant limit~\cite{LewSei-09}.

The conjecture is widely believed to be true in the FQHE-physics community on the grounds that: 

\medskip

\noindent\textbf{1}. It is supported by numerical simulations (numerical diagonalizations of the Hamiltonian for small particle numbers, say up to $\sim 20$, see for example~\cite{Jain-07,Viefers-08} and references therein).

\medskip

\noindent\textbf{2}. Where it to be false, it would be extremely hard to make sense of the experimental data of the FQHE.

\medskip 

It should not actually be necessary to restrict the Hamiltonian to angular momenta below $\ell N(N-1)/2$ to obtain a lower bound to the spectral gap. It is likely that restricting to angular momenta below a larger value (but still of order $N^2$ when $N\to \infty$) would suffice. It is conceivable that the conjecture holds only for small values of $\ell$. If so, a likely threshold~\cite{YanHalRez-01} for the conjecture ceasing to hold is $\ell=7$. 

Finally, there are other versions of the conjecture: for particles living on a sphere or a cylinder instead of in the plane, see~\cite[Sections~3.10 and 3.11]{Jain-07} and references therein.

\bigskip 
\bigskip

\noindent \textbf{Acknowledgments.} Thanks to Elliott H. Lieb, Alessandro Olgiati, Sylvia Serfaty and Jakob Yngvason, collaborations with whom this text is based on. Thanks to F.D.M. Haldane for discussions relating to  Appendix~\ref{app:spectral gap}. Thanks to the European Research Council for funding (under the European Union's Horizon 2020 Research and Innovation Programme, Grant agreement CORFRONMAT No 758620). 

%

\begin{thebibliography}{10}

\bibitem{AndGuiZei-10}
{\sc Anderson, G.~W., Guionnet, A., and Zeitouni, O.}
\newblock {\em An introduction to random matrices}, vol.~118 of {\em Cambridge
  Studies in Advanced Mathematics}.
\newblock Cambridge University Press, Cambridge, 2010.

\bibitem{AroSchWil-84}
{\sc Arovas, S., Schrieffer, J., and Wilczek, F.}
\newblock Fractional statistics and the quantum {H}all effect.
\newblock {\em Phys. Rev. Lett. 53}, 7 (1984), 722--723.

\bibitem{BauBouNikYau-16}
{\sc Bauerschmidt, R., Bourgade, P., Nikula, M., and Yau, H.-T.}
\newblock The two-dimensional {C}oulomb plasma: quasi-free approximation and
  central limit theorem.
\newblock arXiv:1609.08582, 2016.

\bibitem{BauBouNikYau-15}
{\sc Bauerschmidt, R., Bourgade, P., Nikula, M., and Yau, H.-T.}
\newblock Local density for two-dimensional one-component plasma.
\newblock {\em Communications in Mathematical Physics 356}, 1 (2017), 189--230.

\bibitem{BerPap-99}
{\sc Bertsch, G., and Papenbrock, T.}
\newblock Yrast line for weakly interacting trapped bosons.
\newblock {\em Phys. Rev. Lett. 83\/} (1999), 5412--5414.

\bibitem{BurChoTop-15}
{\sc Burchard, A., Choksi, R., and Topaloglu, I.}
\newblock {Nonlocal shape optimization via interactions of attractive and
  repulsive potentials}.
\newblock {\em Indiana Univ. J. Math.\/} (2017).

\bibitem{CheBis-18}
{\sc Chen, Y., and Biswas, R.~R.}
\newblock Gauge-invariant variables reveal the quantum geometry of fractional
  quantum {H}all states.
\newblock arXiv:1807.03306, 2018.

\bibitem{MahaluEtal-97}
{\sc de~Picciotto, R., Reznikov, M., Heiblum, M., Umansky, V., Bunin, G., and
  Mahalu, D.}
\newblock Direct observation of a fractional charge.
\newblock {\em Nature 389\/} (1997), 162--164.

\bibitem{Forrester-10}
{\sc Forrester, P.~J.}
\newblock {\em Log-gases and random matrices}, vol.~34 of {\em London
  Mathematical Society Monographs Series}.
\newblock Princeton University Press, Princeton, NJ, 2010.

\bibitem{FraLie-16}
{\sc Frank, R.~L., and Lieb, E.~H.}
\newblock {A liquid-solid phase transition in a simple model for swarming}.
\newblock {\em Indiana Univ. J. Math.\/} (2017).

\bibitem{Girvin-04}
{\sc Girvin, S.}
\newblock Introduction to the fractional quantum {H}all effect.
\newblock {\em S\'eminaire Poincar\'e 2\/} (2004), 54--74.

\bibitem{GirJac-84}
{\sc Girvin, S., and Jach, T.}
\newblock Formalism for the quantum {H}all effect: {H}ilbert space of analytic
  functions.
\newblock {\em Phys. Rev. B 29}, 10 (1984), 5617--5625.

\bibitem{Goerbig-09}
{\sc Goerbig, M.~O.}
\newblock Quantum {H}all effects.
\newblock arXiv:0909.1998, 2009.

\bibitem{Haldane-83}
{\sc Haldane, F. D.~M.}
\newblock Fractional quantization of the {H}all effect: A hierarchy of
  incompressible quantum fluid states.
\newblock {\em Phys. Rev. Lett. 51\/} (Aug 1983), 605--608.

\bibitem{Haldane-11}
{\sc Haldane, F. D.~M.}
\newblock {Geometrical description of the fractional quantum Hall effect}.
\newblock {\em Phys. Rev. Lett. 107\/} (2011), 116801.

\bibitem{Haldane-18}
{\sc Haldane, F. D.~M.}
\newblock {The origin of holomorphic states in Landau levels from
  non-commutative geometry, and a new formula for their overlaps on the torus}.
\newblock {\em J. Math. Phys. 59\/} (2018), 081901.

\bibitem{Jain-07}
{\sc Jain, J.~K.}
\newblock {\em {Composite fermions}}.
\newblock Cambridge University Press, 2007.

\bibitem{JohetalHal-16}
{\sc Johri, S., Papic, Z., Schmitteckert, P., Bhatt, R.~N., and Haldane, F.
  D.~M.}
\newblock Probing the geometry of the laughlin state.
\newblock {\em New Journal of Physics 18}, 2 (feb 2016), 025011.

\bibitem{Laughlin-83}
{\sc Laughlin, R.~B.}
\newblock Anomalous quantum {H}all effect: An incompressible quantum fluid with
  fractionally charged excitations.
\newblock {\em Phys. Rev. Lett. 50}, 18 (May 1983), 1395--1398.

\bibitem{Laughlin-87}
{\sc Laughlin, R.~B.}
\newblock Elementary theory : the incompressible quantum fluid.
\newblock In {\em The quantum {H}all effect}, R.~E. Prange and S.~E. Girvin,
  Eds. Springer, Heidelberg, 1987.

\bibitem{Laughlin-99}
{\sc Laughlin, R.~B.}
\newblock Nobel lecture: Fractional quantization.
\newblock {\em Rev. Mod. Phys. 71\/} (Jul 1999), 863--874.

\bibitem{Leble-15b}
{\sc Lebl\'e, T.}
\newblock Local microscopic behavior for {2D C}oulomb gases.
\newblock {\em Probability Theory and Related Fields 169}, 3-4 (2017),
  931--976.

\bibitem{LebSer-16}
{\sc Lebl\'e, T., and Serfaty, S.}
\newblock Fluctuations of two-dimensional {C}oulomb gases.
\newblock arXiv:1609.08088, 2016.

\bibitem{LewSei-09}
{\sc Lewin, M., and Seiringer, R.}
\newblock Strongly correlated phases in rapidly rotating {B}ose gases.
\newblock {\em J. Stat. Phys. 137}, 5-6 (Dec 2009), 1040--1062.

\bibitem{LieLos-01}
{\sc Lieb, E.~H., and Loss, M.}
\newblock {\em Analysis}, 2nd~ed., vol.~14 of {\em Graduate Studies in
  Mathematics}.
\newblock American Mathematical Society, Providence, RI, 2001.

\bibitem{LieRouYng-16}
{\sc Lieb, E.~H., Rougerie, N., and Yngvason, J.}
\newblock Rigidity of the {L}aughlin liquid.
\newblock {\em Journal of Statistical Physics 172}, 2 (2018), 544--554.

\bibitem{LieRouYng-17}
{\sc Lieb, E.~H., Rougerie, N., and Yngvason, J.}
\newblock Local incompressibility estimates for the {L}aughlin phase.
\newblock {\em Communications in Mathematical Physics 365}, 2 (2019), 431--470.

\bibitem{LieSeiYng-09}
{\sc Lieb, E.~H., Seiringer, R., and Yngvason, J.}
\newblock {Y}rast line of a rapidly rotating {B}ose gas: {G}ross-{P}itaevskii
  regime.
\newblock {\em Phys. Rev. A 79\/} (2009), 063626.

\bibitem{LunRou-16}
{\sc Lundholm, D., and Rougerie, N.}
\newblock {Emergence of fractional statistics for tracer particles in a
  Laughlin liquid}.
\newblock {\em Phys. Rev. Lett. 116\/} (2016), 170401.

\bibitem{YacobiEtal-04}
{\sc Martin, J., Ilani, S., Verdene, B., Smet, J., Umansky, V., Mahalu, D.,
  Schuh, D., Abstreiter, G., and Yacoby, A.}
\newblock Localization of fractionally charged quasi-particles.
\newblock {\em Science 305\/} (2004), 980--983.

\bibitem{Mehta-04}
{\sc Mehta, M.}
\newblock {\em Random matrices. Third edition}.
\newblock Elsevier/Academic Press, 2004.

\bibitem{OlgRou-19}
{\sc Olgiati, A., and Rougerie}.
\newblock Stability of the laughlin phase against long-range interactions.
\newblock arXiv, 2019.

\bibitem{PapBer-01}
{\sc Papenbrock, T., and Bertsch, G.~F.}
\newblock Rotational spectra of weakly interacting {B}ose-{E}instein
  condensates.
\newblock {\em Phys. Rev. A 63}, 2 (2001), 023616.

\bibitem{Rougerie-INSMI}
{\sc Rougerie, N.}
\newblock {Estimations d'incompressibilit\'e pour la phase de Laughlin}.
\newblock Lettre de l'INSMI, 2015.

\bibitem{Rougerie-hdr}
{\sc Rougerie, N.}
\newblock Some contributions to many-body quantum mathematics.
\newblock arXiv:1607.03833, 2016.
\newblock habilitation thesis.

\bibitem{RouSerYng-13a}
{\sc Rougerie, N., Serfaty, S., and Yngvason, J.}
\newblock Quantum {H}all states of bosons in rotating anharmonic traps.
\newblock {\em Phys. Rev. A 87\/} (Feb 2013), 023618.

\bibitem{RouSerYng-13b}
{\sc Rougerie, N., Serfaty, S., and Yngvason, J.}
\newblock Quantum {H}all phases and plasma analogy in rotating trapped {B}ose
  gases.
\newblock {\em J. Stat. Phys. 154\/} (2014), 2--50.

\bibitem{RouYng-14}
{\sc Rougerie, N., and Yngvason, J.}
\newblock Incompressibility estimates for the {L}aughlin phase.
\newblock {\em Comm. Math. Phys. 336\/} (2015), 1109--1140.

\bibitem{RouYng-15}
{\sc Rougerie, N., and Yngvason, J.}
\newblock Incompressibility estimates for the {L}aughlin phase, part {II}.
\newblock {\em Comm. Math. Phys. 339\/} (2015), 263--277.

\bibitem{RouYng-17}
{\sc Rougerie, N., and Yngvason, J.}
\newblock The {L}aughlin liquid in an external potential.
\newblock {\em Letters in Mathematical Physics 108}, 4 (2018), 1007--1029.

\bibitem{SamGlaJinEti-97}
{\sc Saminadayar, L., Glattli, D.~C., Jin, Y., and Etienne, B.}
\newblock Observation of the $e/3$ fractionally charged {L}aughlin
  quasiparticle.
\newblock {\em Phys. Rev. Lett. 79\/} (Sep 1997), 2526--2529.

\bibitem{Serfaty-15}
{\sc Serfaty, S.}
\newblock {\em {Coulomb Gases and Ginzburg-Landau Vortices}}.
\newblock Zurich Lectures in Advanced Mathematics. Euro. Math. Soc., 2015.

\bibitem{StoTsuGos-99}
{\sc St\"{o}rmer, H., Tsui, D., and Gossard, A.}
\newblock The fractional quantum {H}all effect.
\newblock {\em Rev. Mod. Phys. 71\/} (1999), S298--S305.

\bibitem{Viefers-08}
{\sc Viefers, S.}
\newblock Quantum {H}all physics in rotating {B}ose-{E}instein condensates.
\newblock {\em J. Phys. C 20\/} (2008), 123202.

\bibitem{YanHalRez-01}
{\sc Yang, K., Haldane, F. D.~M., and Rezayi, E.~H.}
\newblock {Wigner crystals in the lowest Landau level at low-filling factors}.
\newblock {\em Phys. Rev. B. 64\/} (2001), 081301(R).

\end{thebibliography}

\end{document}